\documentclass[10pt,a4paper]{IEEEtran}
\usepackage[a4paper, left=1.225cm, right=1.225cm, top=1.85cm, bottom=4.4cm]{geometry}
\IEEEoverridecommandlockouts
\usepackage{xcolor,soul,framed} 
\colorlet{shadecolor}{yellow}
\usepackage[pdftex]{graphicx}
\graphicspath{{../pdf/}{../jpeg/}}
\DeclareGraphicsExtensions{.pdf,.jpeg,.png}

\usepackage[cmex10]{amsmath}
\usepackage{array, cite,xcolor}
\usepackage{mdwmath}
\usepackage{mdwtab}
\usepackage{eqparbox}
\usepackage{url}

\usepackage{xcolor, multicol, tabularx,booktabs, amssymb, dsfont, multirow, cite}
\usepackage[utf8]{inputenc}
\usepackage{amsmath, cases, bbm, graphicx, amsthm, algpseudocode, algorithm, caption, subcaption}

\usepackage{balance}
\usepackage{caption}
\usepackage{xcolor}
\usepackage{mathrsfs}
\usepackage{tikz}
\usetikzlibrary{calc}
\newcommand*\circled[1]{\tikz[baseline=(char.base)]{
    \node[shape=circle, draw, inner sep=1pt, 
        minimum height={\f@size*1.6},] (char) {\vphantom{WAH1g}#1};}}

\usepackage{adjustbox} 

\usepackage{cleveref}

\usepackage{titlesec}

\makeatletter

\newcommand\sfcodefork{%
  \ifnum\the\spacefactor=1000 \expandafter\@firstoftwo\else\expandafter\@secondoftwo\fi
}

\makeatother

\newtheorem{lemma}{Lemma}

\crefname{lemma}{Lemma}{lemmas}

\newcolumntype{C}{>{\centering\arraybackslash}X} 
\newcolumntype{s}{>{\arraybackslash\hsize=.6\columnwidth}X}

\crefname{figure}{Fig.}{Figs.}  
\Crefname{figure}{Fig.}{Figs.}  

\crefname{table}{Table}{Tables}  
\Crefname{table}{Table}{Tables}  

\begin{document}


\bstctlcite{IEEEexample:BSTcontrol}

\title{\huge Learning-Based Collaborative MEC for LLM Inference with Soft-Deadline Awareness via Transformer-Enhanced PPO}



\author{\IEEEauthorblockN{
    Ngoc~Hung~Nguyen and
    Bjorn~Landfeldt \\
}\IEEEauthorblockA{
    Department of Electrical and Information Technology, Lund University, P.O. Box 117, Lund 22100, Sweden
}
\textit{Corresponding author}: Ngoc~Hung~Nguyen, email: ngoc\_hung.nguyen@eit.lth.se\vspace{-10pt}
}

\maketitle

\begingroup
\renewcommand\thefootnote{}
\footnotetext{%
\copyright~2026 IEEE. Personal use of this material is permitted. Permission from IEEE must be obtained for all other uses, in any current or future media, including reprinting/republishing this material for advertising or promotional purposes, creating new
collective works, for resale or redistribution to servers or lists, or reuse of any copyrighted component of this work in other works.%
}
\addtocounter{footnote}{-1}
\endgroup

\begin{abstract}
This paper investigates collaborative mobile edge computing (MEC) servers for large language model (LLM) inference under soft deadline constraints. In this system, to improve the quality of service, computations are expected to be completed within their deadlines. However, due to dependencies among tasks or subtasks, any missed deadline can lead to catastrophic consequences for the entire request. In this context, this work proposes an extended deadline mechanism with constrained flexibility. The main challenges lie in handling large-scale computations under strict latency constraints while limiting the number of allowable deadline extensions, especially in the presence of task dependencies within each request. To tackle these challenges, we develop a transformer-enhanced proximal policy optimization (PPO) framework that enables efficient collaboration among MEC servers. The proposed approach aims to maximize the number of tasks completed within their deadlines while minimizing the use of deadline extensions. By capturing temporal dependencies and cross-server interactions, the transformer improves decision-making for task migration. Simulation results demonstrate that the proposed method significantly outperforms conventional PPO and heuristic-based approaches in terms of task completion rate and overall system efficiency.
\end{abstract}
\IEEEpeerreviewmaketitle
\section{Introduction}
Recently, the rapid development of artificial intelligence (AI) has changed the way people use applications. Instead of using normal search engines, people increasingly tend to raise questions using generative AI (GenAI) applications. The applications not only provide an information summarization but can also process many other types of tasks, such as video creation, image analysis, etc. These services are attractive to a large number of customers but comes at the cost of increased computation for each request or prompt, bringing the system to a high workload state, with large response times. However, the potential application of this AI architecture is wide-ranging in the future, including real-time support (e.g., Avatar-recommendation system) in Virtual and Augmented Reality (XR) systems, or orchestration of intelligent transportation systems in smart cities \cite{qu2025mobile,li2025llm}. Thus, research is ongoing on how the computation of these kinds of tasks gradually becomes more and more crucial in accuracy, latency, and energy consumption \cite{zhang2024edgeshard, husom2025sustainable}.

Mobile edge computing (MEC) provides services close to the edge of the computing system, near end users. However, MEC has limited computational resources compared with central cloud platforms. Nevertheless, MEC has an advantage over cloud platforms in that latencies can be reduced and network bottlenecks can be avoided \cite{mach2017mobile}. However, computing high-workload tasks with complex architectures, such as LLMs, is not easy for a single MEC server. Instead of using a single server to process heavy workloads, many works have proposed collaboration mechanisms among MEC servers. This is expected to become an important direction for improving task computation efficiency in MEC systems \cite{hao2024hybrid,liu2026edge}. 

In fact, once heavy workload applications such as LLMs can be successfully deployed on edge infrastructure, they can enable real-time intelligent services with low latency and enhanced privacy for end devices. In practice, the communication load in both uplink and downlink can be quite heavy, such as video, raw images, or data files. Computation at the edge is expected to reduce latency and improve users' Quality of Service (QoS) experience overall \cite{wang2024end}. Therefore, enabling efficient heavy workload execution on MEC systems is not only beneficial but also necessary for supporting next-generation intelligent applications. Therefore, in this work, we utilize an edge-edge collaboration mechanism to enhance the performance of computation for heavy workload tasks, especially for LLM models \cite{yao2025enhancing}.

Deep reinforcement learning (DRL) has been widely adopted for solving complex decision-making problems, with various advanced algorithms developed in recent years. Among them, proximal policy optimization (PPO) has emerged as a robust and stable policy-gradient method for many practical applications \cite{yang2024beyond}. However, PPO may not consistently achieve high performance in highly dynamic environments due to its sensitivity to hyperparameter settings and its limited capability in capturing long-term temporal dependencies. In the considered MEC system, task migration decisions are influenced not only by the current system state but also by historical information, such as queue evolution, workload fluctuations, and task dependency patterns. Conventional PPO, which relies primarily on instantaneous observations, may lead to myopic decisions under such conditions \cite{hamalainen2020ppo, adkins2024method}. To address this limitation, we incorporate a transformer architecture into the PPO framework. By leveraging self-attention mechanisms, the transformer effectively captures temporal dependencies and cross-server interactions from historical system states, enabling more informed and context-aware decision-making for task migration in collaborative MEC systems.

\begin{enumerate}
    \item We propose a collaborative MEC framework for LLM inference, where each request is modeled as a hierarchical directed acyclic graph (DAG) consisting of multiple tasks and subtasks with dependency constraints. Based on this framework, we formulate a task migration and soft-deadline-aware scheduling problem, and identify its key challenges arising from task dependencies, resource heterogeneity, and latency constraints. The problem is then modeled as a Markov decision process (MDP), capturing system dynamics including queue states, task dependencies, and communication conditions as well as servers' information.

    \item We develop a transformer-enhanced PPO algorithm that exploits temporal dependencies and cross-server interactions to improve decision-making in task migration. We design a normalized system-level reward function that aligns with the optimization objective, enabling stable and effective policy learning.

    \item Simulation results demonstrate that the proposed approach significantly outperforms conventional PPO and static (no migration) task baselines in terms of task completion rate and deadline efficiency.
\end{enumerate}

\begin{figure}[t!]
    \centering       
    \includegraphics[width=0.82\linewidth]{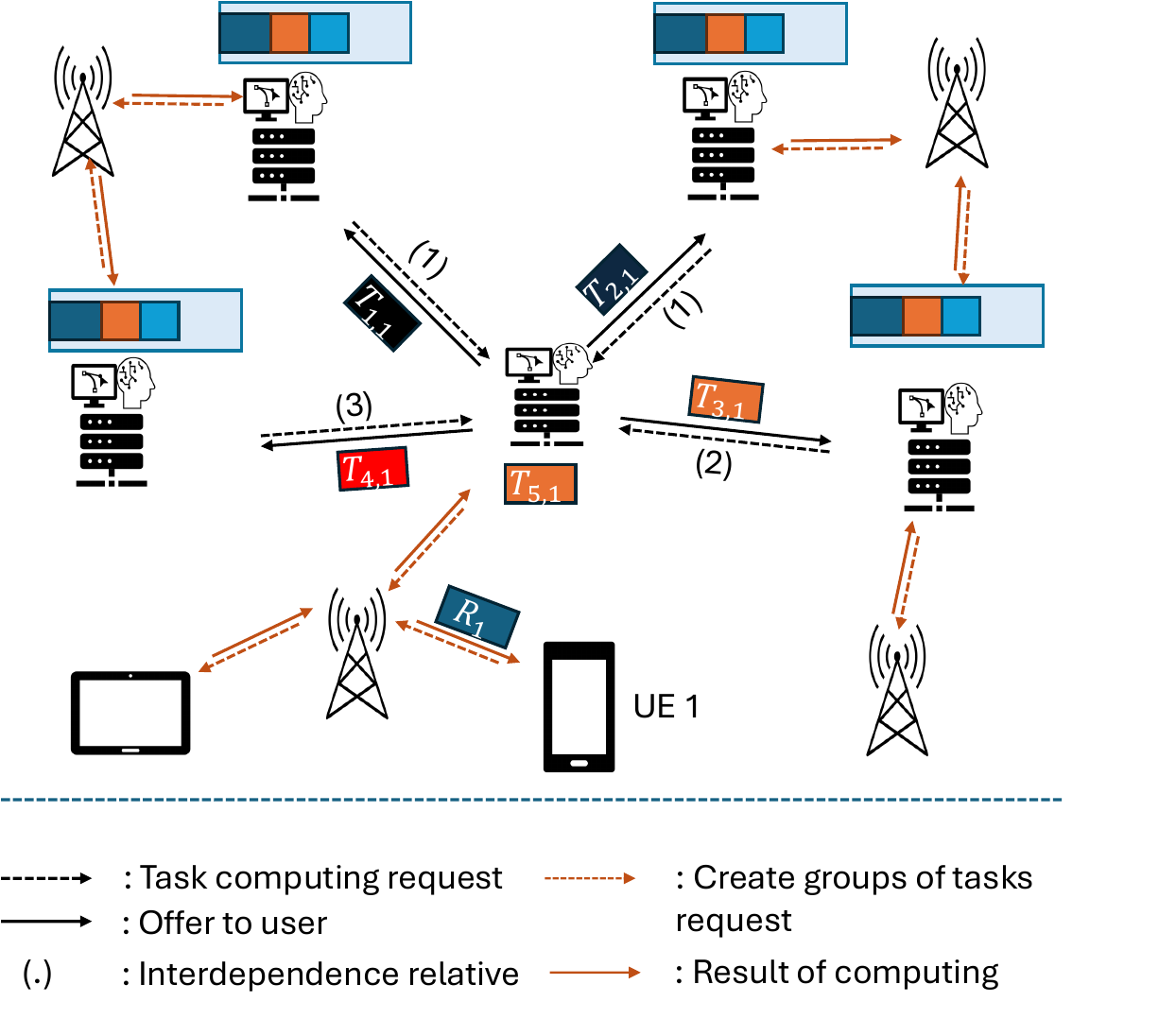}
    \caption{Illustration of system model.}
    \label{fig:systemmodel}
\end{figure}

\section{System Model and Problem Formulation}
\subsection{System Model}
As illustrated in Fig. \ref{fig:systemmodel}, we consider an edge-edge collaboration system expected to enhance the computation capacity of high-demand tasks.
The system includes a set of $|\mathcal{S}|$ MEC servers $\mathcal{S}=\{S_1, S_2, \cdots, S_{|\mathcal{S}|}\}$ and a set of users $\mathcal{U} = \{u_1, u_2, \cdots u_{|\mathcal{U}|}\}$. Each of these servers is provided with a computation capacity of $f_s$ [FLOPs/second]. Note that the system is heterogeneous, meaning that the computational resources of each MEC may vary. Each MEC is distributed in a different region within the deployment area, and each region is serviced by a base station using orthogonal frequency multiple access (OFDMA). The designed system allows for server collaboration through a high-speed physical backhaul by applying remote direct memory access (RDMA). Therefore, tasks in different MEC nodes can be migrated or shared with ultra-low communication latency, enabling real-time cooperative processing across the entire network. Collaboration among MEC servers plays a crucial role in this work, as the system is expected to handle delay-sensitive tasks modeled by DAGs. Let $\mathbf{T}_{k,i}^{\mathrm{[sub]}} = \{T_{k,i}^{(1)}, T_{k,i}^{(2)}, \ldots, T_{k,i}^{(d_{k,i})}\}$ denote the set of subtasks associated with task $T_{k,i}$ in request $\mathcal{R}_i$, where requests arrive according to a Poisson process with rate $\lambda$ [requests/user/s], and $d_{k,i}$ denotes the number of subtasks in $T_{k,i}$. Each task is represented as $T_{k,i} \triangleq \left(\mathbf{T}_{k,i}^{\mathrm{[sub]}}, \mathcal{D}_{k,i}, \alpha_{k,i}, \beta_{k,i}, \gamma_{k,i}\right)$, where $k \in \{1,\ldots,K_i\}$ and $K_i$ denotes the number of tasks in request $\mathcal{R}_i$. Accordingly, request $\mathcal{R}_i$ is represented as $\mathcal{R}_i \triangleq \left(\{T_{k,i}\}_{k=1}^{K_i}, d_i, r_i, a_i\right)$, where $d_i$ and $r_i$ [bits] denote the upload and download data sizes of request $i$, respectively, and $a_i$ denotes the time at which the request arrives in the system. Moreover, $\mathcal{D}_{k,i}$ denotes the set of dependency relationships among the subtasks in $\mathbf{T}_{k,i}^{\mathrm{[sub]}}$, i.e., the edges of the corresponding DAG. For each subtask $j$, we define $T_{k,i}^{(j)} \triangleq \left(\alpha_{k,i}^{(j)}, \beta_{k,i}^{(j)}, \gamma_{k,i}^{(j)}\right)$, where $\alpha_{k,i}^{(j)}$, $\beta_{k,i}^{(j)}$, and $\gamma_{k,i}^{(j)}$ denote the required number of FLOPs, the deadline, and the number of deadline-extension requests of subtask $j$, respectively. Accordingly, the total computational requirement of task $T_{k,i}$ is $\alpha_{k,i} = \sum_{j=1}^{d_{k,i}} \alpha_{k,i}^{(j)}$, its deadline is given by $\beta_{k,i} = \max_{j \in \{1,\ldots,d_{k,i}\}} \beta_{k,i}^{(j)}$, and its total number of deadline-extension requests is $\gamma_{k,i} = \sum_{j=1}^{d_{k,i}} \gamma_{k,i}^{(j)}$. In this work, we aim to maximize the number of tasks completed within their deadlines while minimizing the number of deadline-extension requests.

To enhance computational performance, MEC collaboration is conducted as follows: (1) the host MEC server selects a subset of tasks that require assistance for execution, (2) the host MEC collects system information from neighboring MEC servers to determine appropriate helper servers, (3) the selected helper MEC servers execute the migrated tasks and return the computed results to the host MEC server, and (4) the host MEC server waits for the completion of all tasks within a request, aggregates the results, and delivers the final outcome to the user. 
Note that the system allows deadline extension, where each extension increases the task deadline by an amount equal to its computation time, excluding queuing delay. To mitigate excessive queuing delays, each task is subject to a limited number of deadline extensions for all belonging subtasks.

\subsection{Problem Formulation}
Each user is assumed to move according to a Markovian mobility process, following the model in \cite{nguyen2024deadline}. Let the position of user $u \in [1, \mathcal{U}]$ at time $t$ be denoted by $(x_u(t), y_u(t)) \in \mathcal{A}$, where $\mathcal{A}$ [$\text{m}^2$] represents the considered area. The position of base station $s \in [1, |\mathcal{S}|]$ is denoted by $(x_s, y_s) \in \mathcal{A}$. Accordingly, the distance between user $u$ and base station $s$ at time $t$ is given by $l_{u,s}(t) = \sqrt{(x_u(t) - x_s)^2 + (y_u(t) - y_s)^2}$. The channel coefficient between user $u$ and resource block $n \in [1,N]$ of base station $s$ at time $t$ is expressed as $h_{u,s,n}(t) = \frac{g_{u,s,n}(t)}{l_{u,s}(t)^{p/2}}$, where $p \in [2,4]$ is the path-loss exponent and $g_{u,s,n}(t) \sim \mathcal{CN}(0,1)$ denotes small-scale fading.

Let $x_{u,s,n}(t) \in \{0,1\}$ denote the OFDMA resource allocation variable, where $x_{u,s,n}(t)=1$ if resource block $n$ of base station $s$ is assigned to user $u$ at time $t$, and $x_{u,s,n}(t)=0$ otherwise. Due to the orthogonality of OFDMA, each resource block can be allocated to at most one user at a given time, i.e.,
\begin{align}
    \label{eq:constraint1}
    \sum_{u=1}^{\mathcal{U}} x_{u,s,n}(t) \leq 1, \quad \forall s,n,t.
\end{align}
Moreover, each user can be assigned at most one resource block at each time slot, i.e.,
\begin{align}
    \label{eq:constraint2}
    \sum_{s=1}^{|\mathcal{S}|} \sum_{n=1}^{N} x_{u,s,n}(t) \leq 1, \quad \forall u,t.
\end{align}

The achievable data rate of user $u$ on resource block $n$ of base station $s$ at time $t$ is given by $R_{u,s,n}(t) = B_n \log_2 \left(1 + \frac{P_u \left|h_{u,s,n}(t)\right|^2}{N_0 B_n} \right)$, where $B_n$ denotes the bandwidth of each resource block, $P_u$ is the transmit power of user $u$, and $N_0$ [W/Hz] denotes the additive white Gaussian noise power spectral density. Accordingly, the noise power over each resource block is given by $N_0 B_n$. The total data rate of user $u$ associated with base station $s$ at time $t$ is expressed as
$R_{u,s}(t) = \sum_{n=1}^{N} x_{u,s,n}(t) R_{u,s,n}(t).$
Let $\sigma_i = u$ indicate that request $\mathcal{R}_i$ is generated by user $u$. Then, the total transmission delay associated with request $\mathcal{R}_i$ is given by $\tau_i^{[\mathrm{comm}]}(t) = \frac{d_i + r_i}{R_{u,s}(t)}, \quad \text{with } u = \sigma_i$. Note that, with notation $x_{u,s}(t) \in \{0,1\}$, it implies that $s \in [1, |\mathcal{S}|]$ is selected as the host MEC server for computing task $\mathcal{R}_i$ from user $u$ if $x_{u,s}(t) = 1$ . Thus, only one MEC should be selected as the host MEC, and the resource block should be in the associated base station then:

\begin{align}
    \label{eq:contraintMEChost}
    \sum_{s=1}^{\mathcal{|S|}} x_{u,s}(t) = 1, \text{ s.t } x_{u,s,n}(t) \le x_{u,s}(t), \quad \forall u,s,n,t
\end{align}

For computation, we assume that each MEC server is equipped with $M$ computing resources such that 
\begin{align}
    \label{eq:constraint7}
    \sum_{m=1}^{M} f_{s,m} = f_s,
\end{align}
where $f_{s,m}$ denotes the computation capacity of resource $m$ at server $s$. Each computing resource maintains an independent queue for subtask processing, and all queues operate under the earliest-deadline-first scheduling and non-preemptive policy. At each decision epoch, the system observes all MEC servers and their available computing resources. For load balancing, we assume that once an MEC server $s$ is selected to process a task, its subtasks are assigned to the computing resource with the shortest queue within that server. Let $m^\ast = \arg\min_{m \in [1,M]} |\mathbf{q}_{s,m}|$ denote the computing resource with the shortest queue at server $s$, whose queue is represented by $\mathbf{q}_{s,m^\ast} = \{q_{s,m^\ast}(1), q_{s,m^\ast}(2), \ldots, q_{s,m^\ast}(|\mathbf{q}_{s,m^\ast}|)\}$. Since the queue is dynamically updated as new subtasks are assigned, the position of a subtask may change over time. Let $p^\ast$ denote the final position of subtask $T_{k,i}^{(j)}$ in queue $\mathbf{q}_{s,m^\ast}$. Then, the computation delay of subtask $T_{k,i}^{(j)}$ generated by user $\sigma_i = u$ is given by
$\tau_{k,i}^{[\mathrm{comp}],(j)} = \sum_{p=1}^{p^\ast} \frac{\alpha\big(q_{s,m^\ast}(p)\big)}{f_{s,m^\ast}}$,
where $\alpha\big(q_{s,m^\ast}(p)\big)$ denotes the required number of FLOPs of the subtask at position $p$ in the queue. Thus, $T_{k,i}^{(j)}$ satisfies deadline if:
$a_i + \tau_i^{[\mathrm{comm}]}(t)+ \tau_{k,i}^{[\mathrm{comp}],(j)}(t) \leq \beta_{k,i}^{(j)}(t) + (1-\varphi_{k,i}^{(j)})Y$, here, the deadline can be changed if scheduling fails, thus it depends on the time. $Y$ is a number which is larger than any $\beta_{k,i}^{(j)}$ and $\varphi_{k,i}^{(j)} \in \{0,1\}$ denotes a binary variable indicating whether subtask $T_{k,i}^{(j)}$ is completed before its deadline. 

For the computation dependency condition, a task $T_{k,i}$ is ready for execution only if all its preceding tasks have been completed. Define $\varphi_{k,i}(t) \in \{0,1\}$ as a binary variable indicating whether task $T_{k,i}$ has been completed, where $\varphi_{k,i}(t)=1$ if it is completed and $0$ otherwise. Then, the readiness condition of task $T_{k,i}$ is given by
\begin{align}
    \label{eq:constraint3}
    \prod_{T_{k',i} \in \mathcal{D}_{k,i}} \varphi_{k',i}(t) = 1.
\end{align}

The task $T_{k,i}$ is considered successfully completed if and only if all its subtasks are completed before their deadlines, i.e.,
\begin{align}
    \label{eq:constraint4}
    \sum_{j=1}^{d_{k,i}} \varphi_{k,i}^{(j)}(t) = d_{k,i},
\end{align}
otherwise, the task is considered failed. If a subtask is predicted to miss its deadline during scheduling, the system allows a deadline extension, which is recorded by $\gamma_{k,i}^{(j)}(t)$. In this work, the number of deadline extensions is limited by an upper bound, expressed as
\begin{align}
    \label{eq:constraint5}
    \sum_{j=1}^{d_{k,i}} \gamma_{k,i}^{(j)}(t) \leq \Gamma_{k,i}^{\max},
\end{align}
here, $\Gamma_{k,i}^{\max}$ is the maximum number of deadline extensions for task $T_{k,i}$.

Based on the above system model, we aim to jointly optimize task assignment and deadline extension decisions to maximize the number of successfully completed tasks while minimizing the number of deadline extension requests.

Let $x_{k,i,s}(t) \in \{0,1\}$ denote the task migration variable, where $x_{k,i,s}(t)=1$ if task $T_{k,i}$ is processed at MEC server $s$ at time $t$, and $0$ otherwise. Each task is assigned to exactly one MEC server, i.e.,
\begin{align}
    \label{eq:constraint6}
    \sum_{s=1}^{|\mathcal{S}|} x_{k,i,s}(t) = 1, \quad \forall k,i,t,
\end{align}
which implies that the task can be executed at any MEC among $|\mathcal{S}|$, including the host MEC.

Therefore, the optimization problem is formulated where $\eta$ is a weighting coefficient that controls the importance of minimizing the number of deadline extensions. Let $\mathbf{x} = \{x_{u,s,n}(t),\, x_{u,s}(t),\, x_{k,i,s}(t)\mid \forall u,k,i,s,n,t\}$ 
denote the collection of all binary decision variables, thus the problem formulation can be expressed by:

\begin{align}
\mathbf{P1}: \quad 
\max_{\mathbf{x}} \quad &
\sum_{i}\sum_{k} \varphi_{k,i}(t)
-\eta \sum_{i}\sum_{k}\sum_{j=1}^{d_{k,i}} \gamma^{(j)}_{k,i}(t)\\
\text{s.t.} \quad
& \eqref{eq:constraint1}, \eqref{eq:constraint2},\eqref{eq:contraintMEChost},\eqref{eq:constraint7},\eqref{eq:constraint3},\eqref{eq:constraint4},\eqref{eq:constraint5},\eqref{eq:constraint6}, \tag{9a}\\
& \varphi_{k,i}(t),\;\varphi^{(j)}_{k,i}(t) \in \{0,1\}, \quad \forall k,i,j,t, \tag{9b}\\
& \gamma^{(j)}_{k,i}(t) \in \mathbb{Z}_{+}, \quad \forall k,i,j,t. \tag{9c}
\end{align}
Note that $\varphi_{k,i}(t)$, $\varphi_{k,i}^{(j)}(t)$, and $\gamma_{k,i}^{(j)}(t)$ are outcome variables determined by the task-assignment decision $\mathbf{x}$ and the system dynamics. Therefore, constraints involving these variables implicitly restrict the feasible set of $\mathbf{x}$. In this work, we focus on the optimization of task migration; thus, $x_{u,s,n}(t)$ and $x_{u,s}(t)$ are determined by the predefined random resource-allocation and greedy association policies, respectively, whereas $x_{k,i,s}(t)$ is optimized by the proposed learning agent.
\section{Proposed Method}
\begin{figure}
    \centering
    \includegraphics[width=0.90\columnwidth]{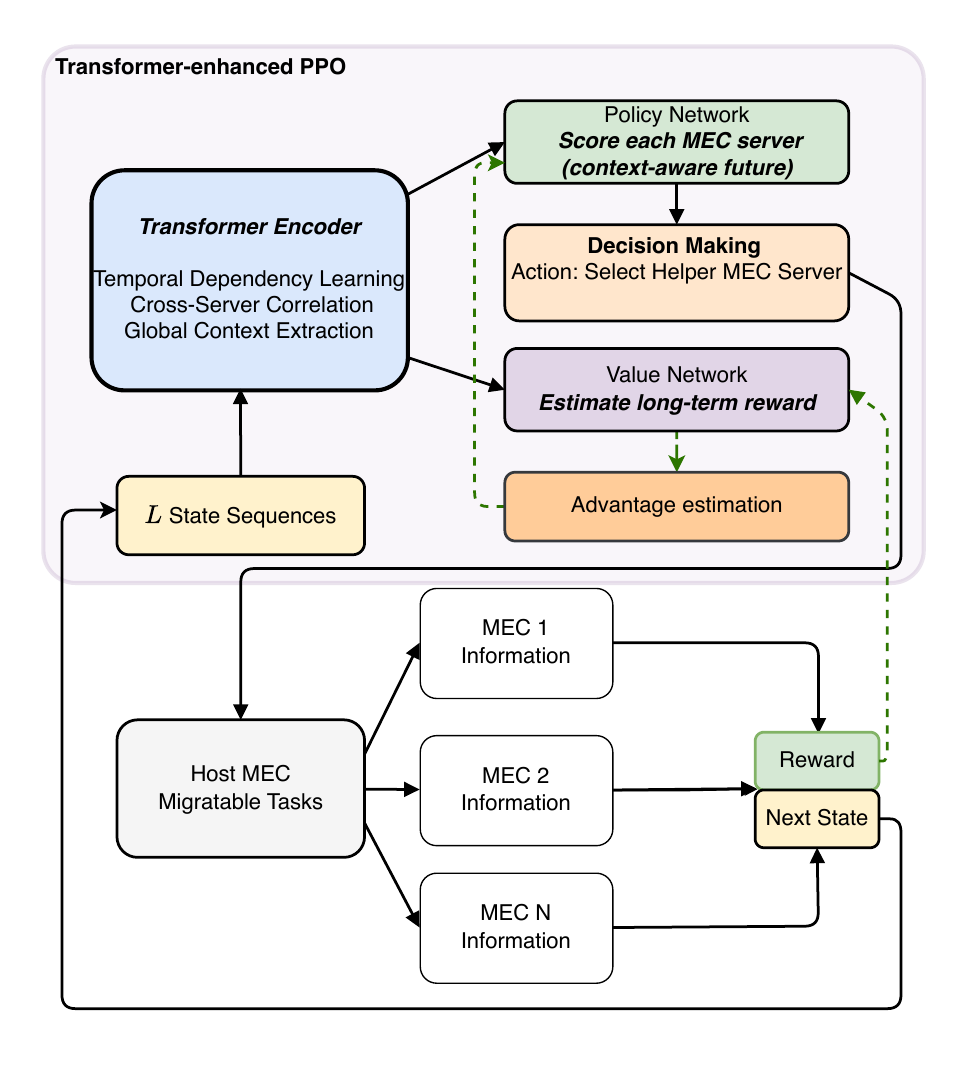}
    \caption{Transformer-enhanced PPO framework for context-aware MEC server selection.}
    \label{fig:proposed_method}
\end{figure}
\vspace{-2mm}
This work focuses on improving the efficiency of collaboration among MEC servers. To simplify the system model, each user is associated with a host MEC server based on a greedy distance-based policy, i.e., the nearest server is selected for task offloading. We assume that each MEC server can broadcast a request for computational assistance at any time instant, where the arrival of such requests follows a general stochastic process. For efficient collaboration, an agent deployed at the host MEC server collects system information from all MEC servers to make appropriate decisions on task migration and resource utilization. In this work, we develop a transformer-based PPO algorithm to efficiently learn the optimal collaboration strategy under dynamic system conditions. The MDP model considered in this work is defined as follows.

\textit{State space:} 
At each decision epoch $t$, the system state is defined as $ \mathbf{s}(t) = \Big\{ \mathcal{X}(t), \mathcal{Q}(t), \mathcal{H}(t), \mathcal{C}(t) \Big\}$, where $\mathcal{X}(t)$ denotes the task-related information, including task dependencies, the number of subtasks $d_{k,i}$, total required resource FLOPs $\alpha_{k,i}$, and deadlines' information $\gamma_{k,i}$, $\beta_{k,i}$, $\mathcal{Q}(t)$ denotes the queue states of all computing resources of all MEC servers and server helper reputation (i.e., computed by the successful tasks per total assigned tasks), $\mathcal{H}(t)$ represents the communication conditions, and $\mathcal{C}(t)$ denotes the computation capacities of MEC servers. To capture temporal correlations, we construct an augmented state by stacking the past $L$ observations $\tilde{\mathbf{s}}(t) = \Big\{ \mathbf{s}(t-L+1), \ldots, \mathbf{s}(t) \Big\}$. This sequence is used as the input to the transformer encoder.

\textit{Transformer Architecture:} The proposed agent employs a time-series Transformer encoder to capture temporal dependencies in the evolution of the MEC system. At each decision epoch, the input to the agent is represented as a sequence $\tilde{\mathbf{s}}(t)\in \mathbb{R}^{L\times F}$, where $F$ denotes the feature dimension of each observation. Each observation vector is first mapped to a latent representation through a linear projection layer. A learnable positional embedding is then added to preserve the temporal order of the observations. The resulting sequence is processed by a Transformer encoder comprising one encoder layer with four self-attention heads, an embedding, and a feed-forward dimension. The encoder employs GELU activation, pre-layer normalization, and no dropout. The contextualized representations generated for all historical observations are aggregated using mean pooling along the temporal dimension, followed by a layer-normalization operation. The resulting global context vector is shared by the actor and critic. Specifically, the actor consists of a linear policy head that produces logits over the candidate supporter MEC servers, whereas the critic consists of a linear value head that estimates the scalar state value. Therefore, the dimension of the action space is equal to the number of available supporter MEC servers. During action selection, a categorical distribution is constructed directly from the policy logits, and the selected action determines the MEC server to which the task is assigned.

\textit{Action space:} 
The action at time $t$ is defined as $a(t) \in \mathcal{S}$, where $a(t)=s$ indicates that the selected MEC server $s$ is chosen to process the task. Therefore, the action corresponds to the task migration decision.

\begin{algorithm}[t]
\caption{Transformer-Enhanced PPO for MEC Collaboration}
\label{alg:transformer_ppo}
\begin{algorithmic}[1]
\State Initialize the MEC environment and the Transformer-enhanced PPO agent.
\For{each decision epoch $t$}
    \State Observe the historical state sequence
    $\tilde{s}(t)=\{s(t-L+1),\ldots,s(t)\}$.
    \State Feed $\tilde{s}(t)$ into the Transformer encoder to obtain the contextual system representation.
    \State Use the policy network to select action $a(t)\in\mathcal{S}$, corresponding to the selected MEC server.
    \State Execute $a(t)$ and observe reward $r(t)$ and next state $s(t+1)$.
    \State Store the current transition for policy updating.
    \If{the policy-update condition is satisfied}
        \State Estimate the advantages using the value network.
        \State Update the policy and value networks using the PPO clipped objective.
        \State Clear the stored transitions.
    \EndIf
\EndFor
\State \Return the trained Transformer-enhanced PPO agent.
\end{algorithmic}
\end{algorithm}

\textit{Reward:} 
In this work, we directly optimize the objective function using deep reinforcement learning (DRL). Therefore, the reward function is carefully designed to be consistent with the objective:
 \begin{align}
    r(t) =
    \frac{
        \omega_1 \sum_{i,k} \varphi_{k,i}(t)
        - \omega_2 \sum_{i,k,j} \gamma_{k,i}^{(j)}(t)
    }{
        \sum_{i,k} \left( d_{k,i} + \Gamma_{k,i} \right)
    },
\end{align}
here $\omega_1$ and $\omega_2$ are the trade-off weights.

\textit{State transition:} 
Given the current state $\mathbf{s}(t)$ and action $a(t)$, the system evolves to the next state $\mathbf{s}(t+1)$ according to the stochastic task arrival process, user mobility, wireless channel variation, and queue dynamics. The transition probability is given by $\Pr\big( \mathbf{s}(t+1) \mid \mathbf{s}(t), a(t) \big)$.

\begin{lemma}
    The formulated decision process satisfies the Markov property and can therefore be modeled as an MDP.
\end{lemma}

\begin{proof}
Let $\mathbf{s}(t)$ and $a(t)$ denote the system state and action at decision epoch $t$, respectively. To prove the Markov property, it is sufficient to show that the conditional distribution of the next state depends only on the current state and action, i.e.,
\begin{align}
    \label{eq:trasition}
    \Pr\big(\mathbf{s}(t+1) \mid \mathcal{H}_t \big) 
    = \Pr\big(\mathbf{s}(t+1) \mid \mathbf{s}(t), a(t) \big),
\end{align}
where $\mathcal{H}_t = \{\mathbf{s}(0), a(0), \ldots, \mathbf{s}(t), a(t)\}$ denotes the system history up to time $t$. By construction, the state $\mathbf{s}(t)$ contains all necessary information for future evolution, including task status, dependency conditions, queue states, communication conditions, and available computation resources. The evolution of the system from time $t$ to $t+1$ is governed by: i) the stochastic task arrival process, ii) the Markovian user mobility model, iii) wireless channel variations, and iv) queue dynamics under the selected action $a(t)$. All these dynamics depend only on the current state $\mathbf{s}(t)$ and action $a(t)$, and are independent of past states and actions given $\mathbf{s}(t)$.

Therefore, the transition probability satisfies \eqref{eq:trasition}, which verifies the Markov property. Hence, the decision process can be modeled as an MDP.
\end{proof}

\begin{table}[t]
\footnotesize
\caption{Simulation and Training Parameters}
\label{tab:system-model-parameters}
\centering
\setlength{\tabcolsep}{4pt}
\begin{tabular}{l c l c}
    \toprule
    \textbf{Parameter} & \textbf{Value} &
    \textbf{Parameter} & \textbf{Value} \\
    \midrule

    \multicolumn{4}{c}{\textit{Simulation Parameters}} \\
    Cellular bandwidth & $10$ MHz &
    Number of RBs $N$ & $10$ \\

    Bandwidth per RB $B_n$ & $1$ MHz &
    User transmit power $P_u$ & $199.526$ mW \\

    Path loss exponent $p$ & $3$ &
    Noise PSD $N_0$ & $-174$ dBm/Hz \\

    \midrule
    \multicolumn{4}{c}{\textit{Transformer Architecture}} \\

    Embedding dimension & $64$ &
    Encoder layers & $1$ \\

    Attention heads & $4$ &
    Feed-forward dimension & $128$ \\

    Dropout & $0$\% &
    Input sequence length & $8$ \\

    \midrule
    \multicolumn{4}{c}{\textit{PPO Training Parameters}} \\

    Learning rate & $10^{-4}$ &
    Discount factor $\gamma$ & $0.99$ \\

    GAE parameter $\lambda$ & $0.95$ &
    PPO clip ratio $\epsilon$ & $0.2$ \\

    Value-loss coefficient & $0.5$ &
    Entropy coefficient & $0.005$ \\

    Maximum gradient norm & $0.5$ \\

    \bottomrule
\end{tabular}
\end{table}

Fig.~\ref{fig:proposed_method} illustrates the overall framework of the proposed transformer-enhanced PPO for MEC collaboration. At each decision epoch, the agent observes a sequence of past system states, denoted as $\{ \mathbf{s}(t-L+1), \ldots, \mathbf{s}(t) \}$, which captures the temporal evolution of queue dynamics, task dependencies, and wireless conditions. This historical state sequence is fed into a transformer encoder to extract context-aware representations by modeling both temporal dependencies and cross-server correlations.

The output of the transformer is then utilized by the policy and value networks in the PPO framework. Specifically, the policy network (actor) scores each MEC server and selects the most suitable helper server for task migration, while the value network (critic) estimates the long-term expected reward based on the current system context. The selected action is executed in the collaborative MEC environment, where the host MEC offloads tasks that can migrate to candidate helper MEC servers.
\begin{table*}[t]
\centering
\caption{Subtasks' completion rate (Comp.) and remaining adjustment (Rem.) rates across operations and methods.}
\label{tab:comparision}
\scriptsize
\setlength{\tabcolsep}{0.60pt}  
\renewcommand{\arraystretch}{1.0}
\begin{tabular}{c|ccc|ccc|ccc|ccc|ccc|ccc|ccc|ccc}
\toprule
\textbf{Op.}
& \multicolumn{6}{c|}{feed\_forward}
& \multicolumn{6}{c|}{self\_attention}
& \multicolumn{6}{c|}{layer\_norm}
& \multicolumn{6}{c}{residual\_add} \\
\cmidrule(lr){1-7} \cmidrule(lr){8-13} \cmidrule(lr){14-19} \cmidrule(lr){20-25}
\multirow{3}{*}{\textbf{Users}} 
& \multicolumn{3}{c|}{Comp.} & \multicolumn{3}{c|}{Rem.}
& \multicolumn{3}{c|}{Comp.} & \multicolumn{3}{c|}{Rem.}
& \multicolumn{3}{c|}{Comp.} & \multicolumn{3}{c|}{Rem.}
& \multicolumn{3}{c|}{Comp.} & \multicolumn{3}{c}{Rem.} \\
\cmidrule(lr){2-4} \cmidrule(lr){5-7}
\cmidrule(lr){8-10} \cmidrule(lr){11-13}
\cmidrule(lr){14-16} \cmidrule(lr){17-19}
\cmidrule(lr){20-22} \cmidrule(lr){23-25}
& Trans. & Trad. & No-Migr.
& Trans. & Trad. & No-Migr.
& Trans. & Trad. & No-Migr.
& Trans. & Trad. & No-Migr.
& Trans. & Trad. & No-Migr.
& Trans. & Trad. & No-Migr.
& Trans. & Trad. & No-Migr.
& Trans. & Trad. & No-Migr. \\
\midrule
10 & \textbf{0.91} & 0.90 & 0.90 & \textbf{0.91} & 0.91 & 0.89
   & \textbf{0.97} & \textbf{0.97} & \textbf{0.97} & \textbf{0.91} & \textbf{0.91} & 0.89
   & \textbf{1.00} & 0.99 & 0.99 & \textbf{0.96} & \textbf{0.96} & 0.95
   & \textbf{0.91 }& \textbf{0.91} & 0.90 & \textbf{0.91} & \textbf{0.91} & 0.89 \\
20 & \textbf{0.64} & 0.63 & 0.55 & 0.70 & \textbf{0.71} & 0.58
   & \textbf{0.88} & 0.87 & 0.83 & \textbf{0.70} & \textbf{0.70} & 0.59
   & \textbf{0.98} & 0.97 & \textbf{0.98} & \textbf{0.85} & \textbf{0.85} & 0.81
   & \textbf{0.66} & 0.65 & 0.57 & \textbf{0.72} & \textbf{0.72} & 0.59 \\
30 & \textbf{0.39} & 0.38 & 0.36 & 0.59 & \textbf{0.60} & 0.49
   & \textbf{0.70} & 0.68 & 0.67 & 0.54 & \textbf{0.56} & 0.48
   & \textbf{0.96} & 0.92 & \textbf{0.96} & 0.77 & \textbf{0.79} & 0.76
   & \textbf{0.42} & 0.41 & 0.39 & 0.60 & \textbf{0.62} & 0.51 \\
40 & \textbf{0.31} & 0.31 & 0.23 & \textbf{0.51} & 0.49 & 0.38
   & \textbf{0.55} & 0.54 & 0.46 & \textbf{0.44} & 0.42 & 0.33
   & \textbf{0.97} & \textbf{0.97} &\textbf{0.97} & \textbf{0.76} & \textbf{0.76} & 0.72
   & \textbf{0.33} & 0.32 & 0.25 & \textbf{0.52} & 0.5 & 0.40 \\
50 & \textbf{0.26} & \textbf{0.26} & 0.25 & \textbf{0.46} & 0.45 & 0.35
   & \textbf{0.49} & \textbf{0.49} & 0.45 & \textbf{0.40} & 0.39 & 0.33
   & \textbf{0.95} & \textbf{0.95} & \textbf{0.95} & \textbf{0.75} & \textbf{0.75} & 0.71
   & \textbf{0.28} & \textbf{0.28} & 0.26 & \textbf{0.46} & \textbf{0.46} & 0.36 \\
\bottomrule
\end{tabular}
\end{table*}
After task execution, the environment returns the reward and the updated system state, which are used to train the critic and compute the advantage function for policy updates. This closed-loop interaction enables the agent to learn an efficient collaboration strategy by jointly considering system dynamics, resource availability, and task requirements. By leveraging the transformer, the proposed method effectively captures long-term temporal patterns and inter-server relationships, leading to more informed and robust decision-making compared to conventional PPO approaches.

\section{Numerical Study}
\label{sec:numerical_simulation}

\subsection{System Parameters and Baselines}
The simulated MEC system consists of six servers, each equipped with a heterogeneous number of GPU cards randomly selected from $\{1, 2, 4, 8, 12\}$. The computational capability of each GPU card is set to $3.74 \times 10^{12}$ FLOPs, which is representative of current high-performance accelerators. The network coverage is modeled as a two-dimensional square area of size $1000 \times 1000$ [$m^2$]. User devices are uniformly distributed within this region, with the number of users ranging from 10 to 50 per server on average.


Each user generates inference requests associated with LLM workloads following a GPT-3-like architecture. Each subtask corresponds to a component in the model, including layer normalization, multi-head self-attention, feed-forward networks, and residual addition. The prompt length is uniformly distributed between 30 and 50 tokens, while the corresponding transmission data size ranges from 100 to 1000 [KB].
The computational cost of each transformer component varies significantly: multi-head self-attention (with 96 heads) requires approximately $3.2\times10^{8}$ to $1.2\times10^{9}$ FLOPs, followed by feed-forward networks with $2.0\times10^{8}$ to $3.3\times10^{8}$ FLOPs. In contrast, layer normalization and residual addition incur much lower costs, on the order of $10^{4}$ to $10^{5}$ FLOPs. These values are derived based on standard FLOPs analysis of multi-head self-attention and feed-forward networks \cite{vaswani2017attention, shoeybi2019megatron}, and reflect practical implementation variations across different input sizes and system configurations.
Additional system and communication parameters are summarized in Table~\ref{tab:system-model-parameters}. These include wireless transmission settings, channel characteristics, and infrastructure configurations.



\begin{figure}[t]
\centering
\begin{subfigure}[t]{0.85\columnwidth}
\centering
\includegraphics[width=\linewidth]{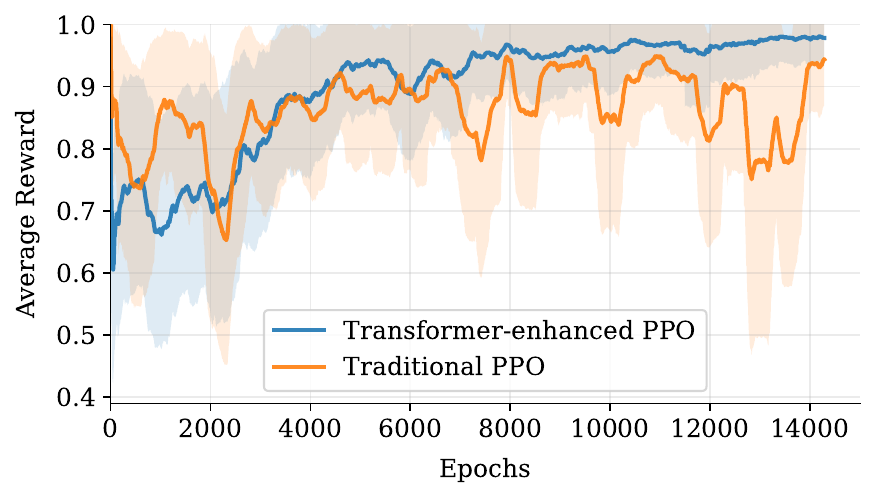}
\caption{Training performance of the proposed method.}
\label{fig:combined_results_1}
\end{subfigure}
\hfill
\begin{subfigure}[t]{0.85\columnwidth}
\centering
\includegraphics[width=\linewidth]{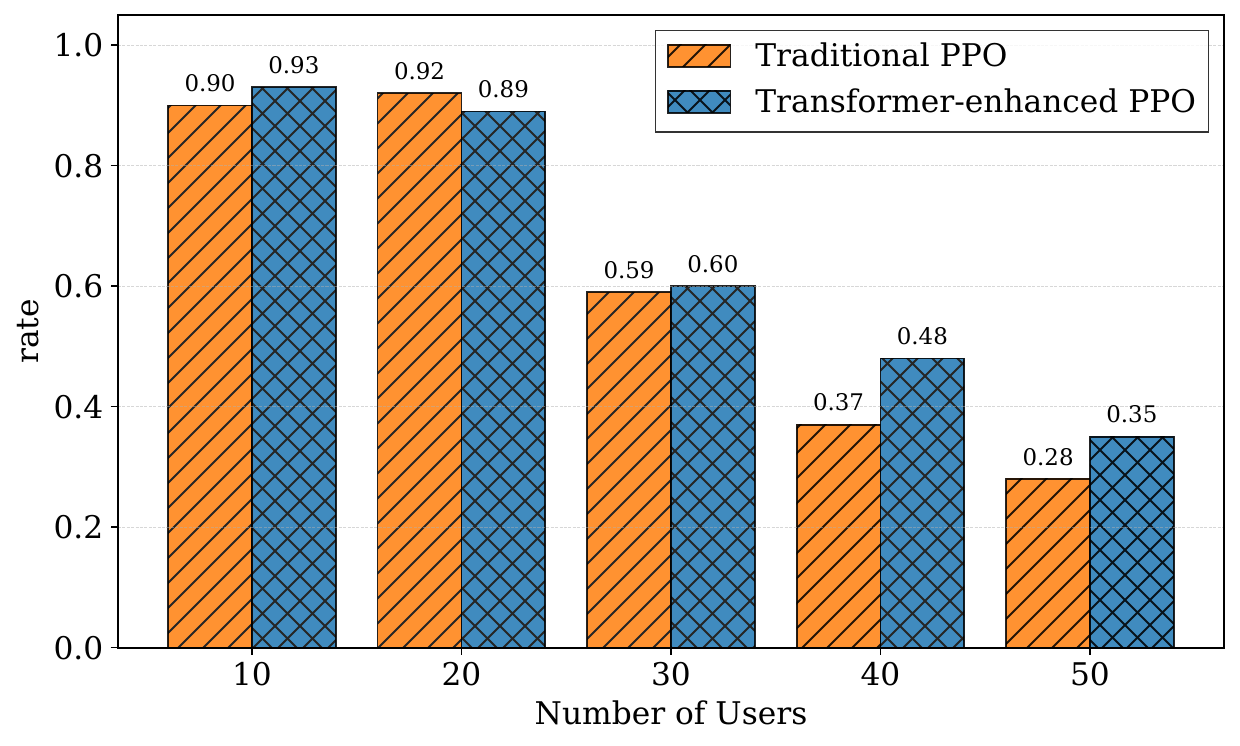}
\caption{Server completion tasks' rate versus users.}
\label{fig:combined_results_2}
\end{subfigure}
\caption{Performance comparison of training and task completion.}

\end{figure}




\subsection{Performance Comparison}
In order to demonstrate the performance of the proposed DRL approach with the system, we design a traditional PPO agent without a transformer encoder component (Traditional PPO). Furthermore, the proposed DRL method is also compared with the system where the migration is not applied (No migration).
\subsubsection{Convergence of Transformer-enhanced PPO} As shown in Fig.~\ref{fig:combined_results_1}, the convergence of the proposed DRL is slower in the first 4000 epochs compared to the traditional PPO, which converges in less than 2000 epochs. However, the average reward curve of the proposed DRL is consistent and exhibits less fluctuation than the baseline scheme. At the end of the training epochs, the proposed DRL's convergence shows a significant gap with the others. The slower convergence of the proposed method is mainly due to the increased model complexity and the use of historical state sequences, which enlarges the learning space. However, the transformer enables the agent to capture temporal dependencies and helps reduce variance through its attention mechanisms, resulting in more stable training behavior and improved long-term performance. In another aspect, we compare the completion rate of migrated tasks in Fig.~\ref{fig:combined_results_2}, the result shows that the proposed transformer-enhanced PPO has consistent and significant gaps compared to the other two algorithms. This shows the  credibility of the proposed algorithm, identical to the results shown in the learning curve. 
\subsubsection{The completion rate of migration tasks} Fig.~\ref{fig:completion_rate} shows the completion rates of the different methods, which directly impact overall task completion performance. Under different operators, two key characteristics are observed. First, the migration policy increases both the number of subtasks completed during computation and the number of remaining deadline adjustment opportunities in the system. Second, the Transformer-enhanced PPO improves the total number of completed subtasks by 2\%, as shown in Fig.~\ref{fig:completion_rate}, which is accompanied by a 9\% increase in migrated-task completion performance, as shown in Fig.~\ref{fig:combined_results_2}. Note that these results are obtained with 40 users, where even a 1\% improvement can correspond to millions of additional successfully completed subtasks.
The completion rate and the remaining number of deadline extensions for subtasks under varying numbers of users are reported in Table~\ref{tab:comparision}. As the number of users increases, the completion rate consistently decreases across most operations, indicating reduced system effectiveness under higher load. In another aspect, simpler operations such as layer normalization remain highly stable, while more complex components like self-attention and feed-forward networks exhibit significant performance degradation. In addition, the transformer-enhanced PPO maintains higher completion rates and remaining adjustments compared to baseline methods, demonstrating better robustness under scaling.
\begin{figure}
\centering
\begin{subfigure}[t]{0.85\columnwidth}
\centering
\includegraphics[width=\linewidth]{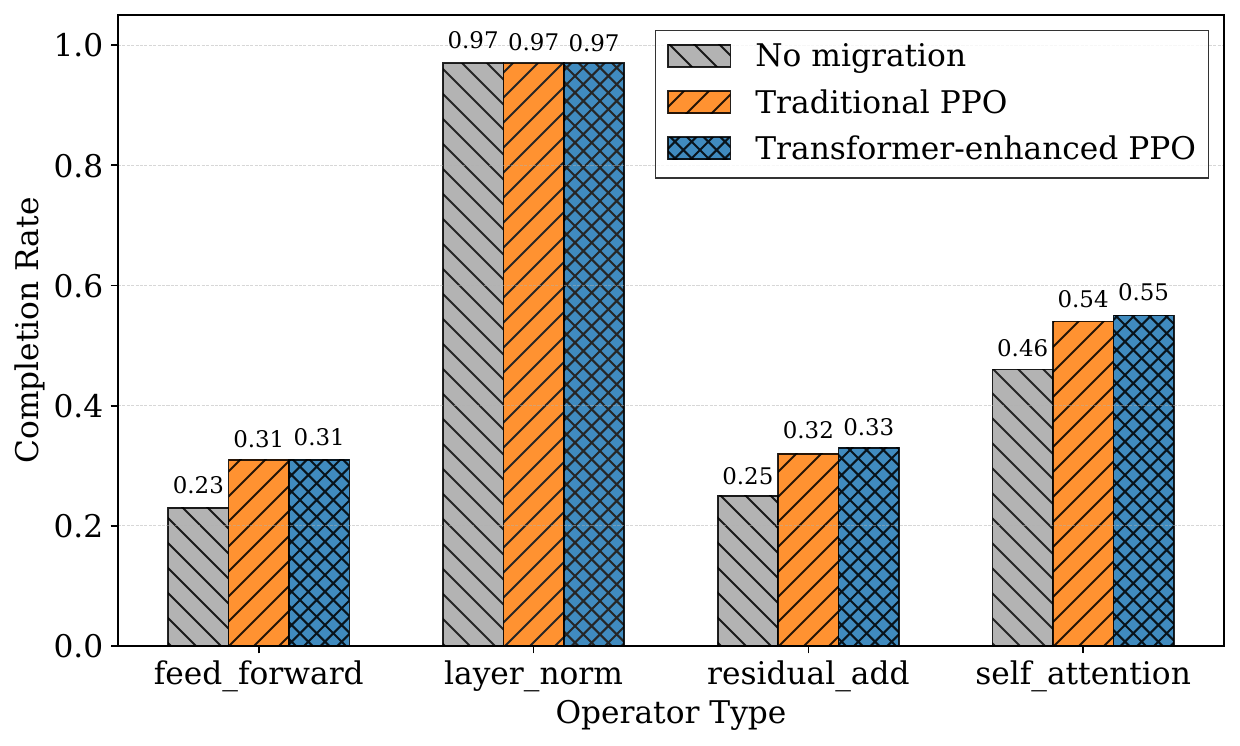}
\caption{Completion rate of subtasks}
\label{fig:completion_rate}
\end{subfigure}
\hfill
\begin{subfigure}[t]{0.85\columnwidth}
    \centering
    \includegraphics[width=\linewidth]{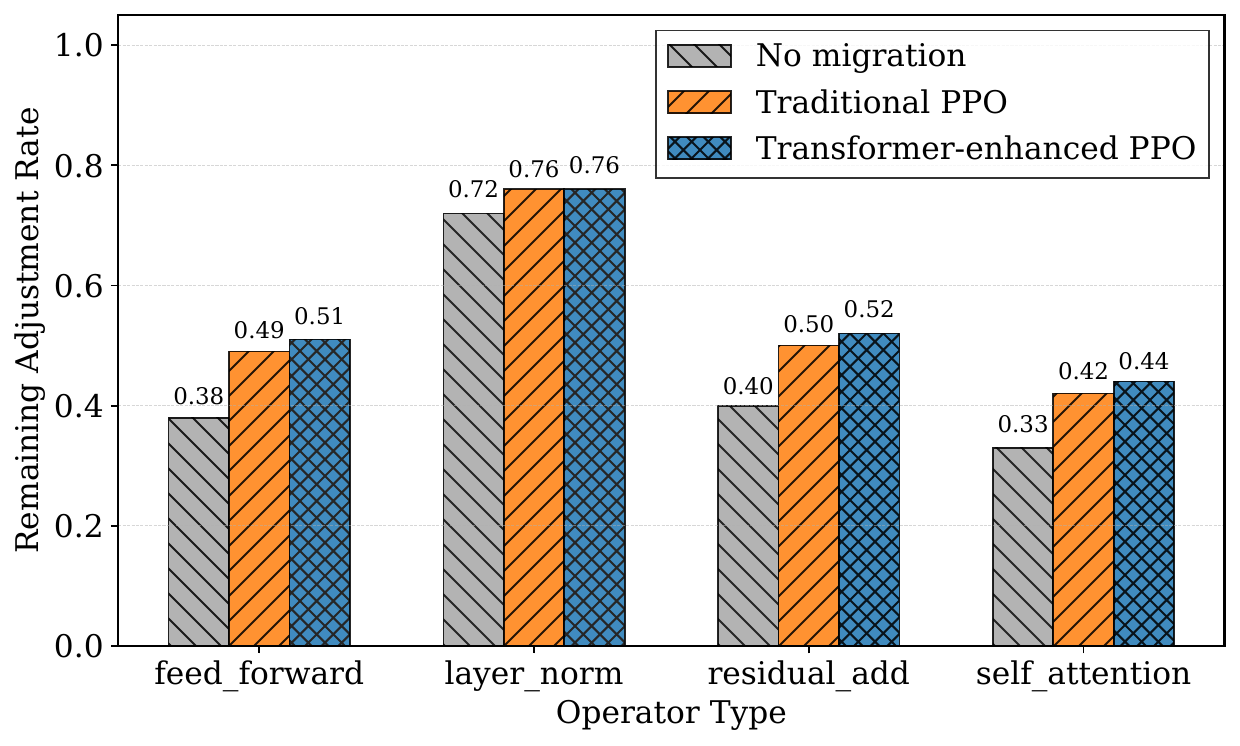}
    \caption{Remaining adjustment rate.}
    \label{fig:remaining_rate}
\end{subfigure}

\caption{Completion rate and remaining adjustment across operator types (40 users).}
\end{figure}
\section{Conclusion}
\label{sec:cons}
This work has proposed a learning-based collaborative MEC framework for LLM inference under soft-deadline constraints. Transformer-enhanced PPO approach effectively captures temporal system information as well as cross-server interactions, enabling more informed task-migration decisions. The results demonstrate that the proposed method achieves higher task completion rates with fewer deadline extensions compared to conventional approaches. In the future, we will extend the framework by incorporating adaptive user association and more realistic system dynamics, as well as exploring scalable multi-agent learning schemes for large-scale MEC deployments. In addition, integrating resource-aware LLM model compression and dynamic workload partitioning remains a promising direction to further improve system efficiency. Furthermore, the current framework assumes profiled computational workloads and a high-speed inter-MEC backhaul. Future work will consider runtime variability in LLM inference, bursty traffic, backhaul congestion, and extension-duration-aware objectives.
\section*{Acknowledgment}
This work was supported by the Swedish Government Agency for Innovation Systems (Vinnova) through the 6G-FOX project under Grant No. 2024-02443, and by ELLIIT, the Linköping-Lund initiative on IT and mobile communication.

\bibliographystyle{./trans/IEEEtran}
\bibliography{ ./trans/IEEEabrv, refs}
\end{document}